\newcommand \ignore[1]{}
\documentclass[runningheads]{llncs}
\spnewtheorem{assumption}[theorem]{Assumption}{\bfseries}{\itshape}
\usepackage[T1]{fontenc}
\usepackage{graphicx}
\usepackage{enumitem}
\usepackage{amsmath}
\usepackage{amssymb}
\usepackage[section]{placeins}
\usepackage{booktabs}
\usepackage{tabularx}
\usepackage{float}
\usepackage{tikz}
\usetikzlibrary{arrows.meta}
\begin{document}
\title{Quantum-Resilient Decentralized AI Economies:\\
       Proof-of-Useful-Work and Post-Quantum Security}
\titlerunning{Quantum-Resilient Decentralized AI Economies}
\authorrunning{C. Barbaccia}
\author{Connor Barbaccia\inst{1}, Sudip Vhaduri\inst{1}, Sayanton Dibbo\inst{1}\thanks{corresponding author}}

\institute{
University of Alabama, Tuscaloosa, AL 35487, USA\\
\email{crbarbaccia@crimson.ua.edu, \{svhaduri,sdibbo\}@ua.edu}
}

%







\maketitle

\begin{abstract}
Proof-of-Work blockchains secure consensus through hash puzzles, producing no external value. In this research, we propose a decentralized AI economy where nodes are rewarded for useful machine-learning work, i.e., inference and training, instead of ineffective hashing method. Our proposed three-layer architecture separates compute, validation, and economic coordination. We formalize it via a $(\theta_c,\theta_w,W)$-closed-loop token economy and derive a sufficient-stake condition for honest participation. While existing Grover's algorithm provides only a quadratic speedup against hash puzzles, it does not accelerate ML-native linear algebra. On the other hand, Shor's algorithm threatens classical blockchain signatures. Post-quantum migration to lattice-based and hash-based standards can address the signature layer. Therefore, useful-work consensus thus offers both economic and quantum-security advantages over classical proof-of-work.

\keywords{Decentralized AI \and Proof of useful work \and Tokenomics \and
Verifiable inference \and zkML \and Post-quantum blockchain \and
Grover's algorithm \and Shor's algorithm}
\end{abstract}
\section{Introduction}
\label{sec:intro}

The two dominant computational trends of the past decade, large machine-learning models and blockchain-based coordination, have largely developed in isolation.
AI workloads are trained and served by a handful of organizations whose compute and energy demands have grown quickly~\cite{patterson2021carbon}, while blockchain networks demonstrate that globally distributed participants
can coordinate computation without trusted intermediaries.
Yet the computation that secures these networks, i.e., SHA-256 hashing in Bitcoin---has no utility beyond the chain itself~\cite{nakamoto2008bitcoin}, drawing an estimated $13.39$\,GW of electricity and producing roughly $65.4$\,MtCO$_2$ annually~\cite{devries2022revisiting}.
Decentralized physical infrastructure (DePIN) networks such as Akash, Render, io.net, and Bittensor already route GPU compute through tokenized incentives~\cite{akash2018whitepaper,render2023whitepaper,ionet2026docs,%
rao2020bittensor}, but the token value remains weakly coupled to the metered
service demand in most of these systems~\cite{mafrur2025ai}.

This paper asks two complementary questions.
\begin{quote}
\emph{Can decentralized networks coordinate useful AI computation in
place of wasteful hashing, while preserving trust and incentive
properties---and how does this architectural choice affect quantum
resilience?}
\end{quote}
The second question is timely: rapid progress in quantum hardware has made the post-quantum risks of blockchain systems a pressing concern~\cite{fernandez2020towards,ghosh2025quantum}, and a useful-work system must understand whether it inherits or avoids those risks.

The contributions of this paper are:
\begin{itemize}[leftmargin=*]
  \item A \emph{three-layer reference architecture} (compute, validation,
        economic coordination) for decentralized AI networks.
  \item A formal \emph{$(\theta_c,\theta_w,W)$-closed-loop token economy}
        (Definition~\ref{def:loop}), making closed-loop operation a precise
        evaluation criterion rather than a design goal.
  \item A \emph{sufficient-stake bound} for honest participation
        (Proposition~\ref{prop:ic}) derived from a concrete reward function
        under an explicit threat model.
  \item A \emph{layer-by-layer quantum-resilience analysis} showing that
        ML-native work layers reduce exposure to Grover-style speedups,
        while Shor-algorithm risks require post-quantum signature migration
        at the transaction layer.
\end{itemize}

\section{Background: Proof of Work and Useful Work}
\label{sec:background}

In Table~\ref{tab:notation}, we summarize the main notation used in the formal model and incentive analysis.

\begin{table}
\caption{Notation used throughout the paper.}
\label{tab:notation}
\centering
\small
\renewcommand{\arraystretch}{0.88}
\begin{tabular}{p{0.18\textwidth}p{0.74\textwidth}}
\toprule
Symbol & Meaning \\
\midrule
$\tau$ & Native token of the decentralized AI network. \\
$\mathcal{N}$ & Decentralized AI network. \\
$n$ & Nonce tested in a proof-of-work puzzle. \\
$H(\cdot)$ & Hash function. \\
$T$ & Difficulty target in proof of work. \\
$\|$ & Concatenation operator. \\
$t$ & Time period being measured. \\
$W$ & Rolling measurement window for closed-loop demand. \\
$\theta_c$ & Minimum required share of AI-service payments made in $\tau$. \\
$\theta_w$ & Minimum required share of new token issuance allocated to validated useful work. \\
$C_\tau(t)$ & Fraction of aggregate AI-service payments made in $\tau$ over window $W$. \\
$I_w(t)$ & Fraction of newly issued $\tau$ supply allocated according to validated productive work. \\
$f$ & Fraction of validator stake controlled by a Byzantine adversary. \\
$x$ & Task input or model input. \\
$y$ & Submitted output for task $x$. \\
$f_\theta$ & Model function with parameters $\theta$. \\
$\theta$ & Model parameter vector. \\
$\theta_t$ & Model parameters at training step $t$. \\
$\theta_{t+1}$ & Model parameters after one gradient update. \\
$\eta$ & Learning rate. \\
$\hat{g}_t$ & Estimated gradient at step $t$. \\
$k$ & Number of replicated executions or validator ratings. \\
$h$ & Probability that an individual replica is honest. \\
$Q(y,x)$ & Validator-aggregated quality score for output $y$ on task $x$. \\
$C(y)$ & Verifiable compute-cost proxy for producing output $y$. \\
$P(y,x)$ & Slashing penalty term for output $y$ on task $x$. \\
$S$ & Staked bond posted by a compute node. \\
$\alpha$ & Protocol reward rate per unit of validated quality. \\
$\beta$ & Cost coefficient converting $C(y)$ into expected node cost. \\
$\gamma$ & Fraction of the bonded stake forfeited after a successful challenge. \\
$q_H$ & Expected validator quality score under honest execution. \\
$q_L$ & Expected validator quality score under low-effort or dishonest execution. \\
$c_H$ & Compute cost of honest execution. \\
$c_L$ & Compute cost of the cheaper deviation. \\
$p$ & Probability that a deviation is detected and successfully challenged. \\
$N_{\mathrm{nodes}}$ & Number of compute nodes in the worked example. \\
$c_{\mathrm{gpu}}$ & GPU cost per GPU-hour. \\
$\lambda_{\mathrm{avg}}$ & Average query rate per node. \\
$\Lambda$ & Total network query volume per day. \\
$p_q$ & User price per query. \\
$\epsilon$ & Numerical tolerance used for tolerance-based output comparison. \\
\bottomrule
\end{tabular}
\renewcommand{\arraystretch}{1.0}
\end{table}

\subsection{Proof of Work and Its Discontents}

In Bitcoin~\cite{nakamoto2008bitcoin}, miners search for a nonce $n$ such
that
\begin{equation}
  H(\text{block} \,\|\, n) < T,
  \label{eq:pow-condition}
\end{equation}
where $H(\cdot)$ is the hash function, $\|$ denotes concatenation, and $T$
is the difficulty target.
The puzzle is cheap to verify but costly to solve, and security scales
linearly with mining power.
This mechanism secures the ledger but produces no reusable external output.

\emph{Proof of Useful Work} (PoUW) replaces hash puzzles with tasks that
provide external benefit while retaining verifiability and Sybil
resistance~\cite{ball2018proofs}.
Proposed tasks have included matrix multiplication, protein folding, and,
more recently, ML workloads.
Fig.~\ref{fig:evolution} traces this progression.

\begin{figure}[H]
\centering
\includegraphics[width=\columnwidth]{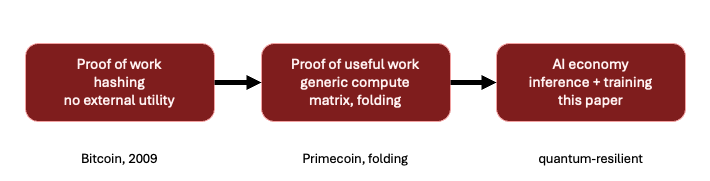}
\caption{Evolution of \emph{work} in decentralized networks: from
hash puzzles (no external value) to generic useful tasks to AI-native
computation as the unit of work.}
\label{fig:evolution}
\end{figure}
\FloatBarrier

\subsection{Prior Work}

Related work spans blockchain--ML integration
surveys~\cite{blockchainml2024,dinh2018aiblockchain,geren2025blockchain};
zkML systems~\cite{kang2022scaling,chen2024zkml,abbaszadeh2024zero};
optimistic verification~\cite{conway2024opml};
reputation-based inference networks~\cite{bouchiha2024llmchain};
communication-efficient distributed
training~\cite{douillard2023diloco,ryabinin2023swarm,lin2018deep};
Byzantine-tolerant
aggregation~\cite{blanchard2017machine,yin2018byzantine}; and post-quantum
blockchain security~\cite{fernandez2020towards,ghosh2025quantum}.

Active decentralized AI systems include
\textbf{Bittensor}~\cite{rao2020bittensor} (peer-ranked intelligence
markets), \textbf{Gensyn}~\cite{gensyn2022} (decentralized training with
proof-of-learning validation), and
\textbf{Akash}~\cite{akash2018whitepaper} (GPU marketplace with
client-defined correctness).
None of these systems fully closes the economic loop defined in
Section~\ref{sec:system}, and none provides a layer-by-layer quantum
threat analysis.

\section{System Model: A Closed-Loop AI Economy}
\label{sec:system}

We propose a three-layer reference architecture, illustrated in
Fig.~\ref{fig:architecture}.
The \emph{compute layer} performs AI work; the \emph{validation layer}
checks whether that work was done correctly; and the \emph{economic layer}
links validated work, token payments, and node compensation.

\begin{figure}[H]
\centering
\includegraphics[width=\columnwidth]{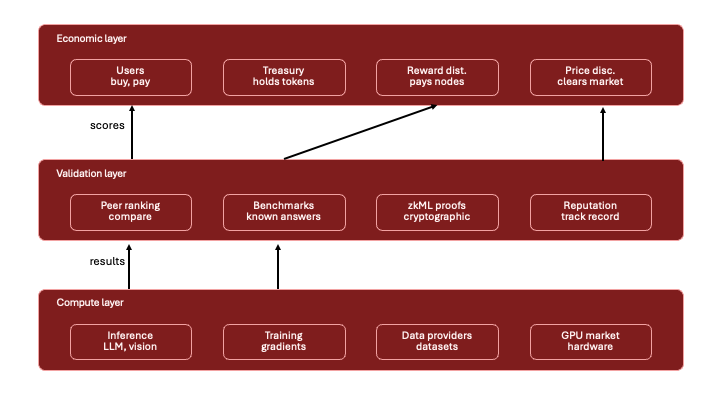}
\caption{Three-layer reference architecture for a decentralized AI economy. Compute nodes perform inference and training, validators score outputs,
and economic layer converts validated work into token rewards. Work and scores flow upward, whereas user payments flow downward to compensate nodes.}
\label{fig:architecture}
\end{figure}
\FloatBarrier

\begin{assumption}[Threat model and network]
\label{assm:threat}
We work in a partially synchronous network as defined by Dwork, Lynch, and
Stockmeyer~\cite{dwork1988consensus}: messages may be delayed arbitrarily
before an unknown Global Stabilization Time, but arrive within a known
bound afterward.
The Byzantine adversary controls fewer than one-third of validator stake
($f < 1/3$), following standard BFT
consensus~\cite{castro1999practical,buchman2016tendermint}.
Compute nodes are rational-majority actors following the protocol only when
it is their dominant strategy.
Cryptographic adversaries are polynomial-time; a public model registry
allows the network to verify model identity independently of output
correctness.
\end{assumption}

All later claims about incentive compatibility, verification soundness, and
liveness hold under Assumption~\ref{assm:threat} unless stated otherwise.

\subsection{Compute Layer}

Nodes contribute GPUs or specialized accelerators for two task types. \emph{Inference} applies a hosted model to a user request: stateless, low-latency, and high-throughput. \emph{Training} updates model parameters through gradient computation, which is stateful, bandwidth-intensive, and harder to coordinate because workers must stay aligned on shared model state. Each node advertises available models, hardware, and prices; a scheduler matches incoming tasks to suitable nodes.

\subsection{Validation Layer}


For all ML workloads, no single verification technique is adequate. The layer combines: \emph{cryptographic verification} via zkML proofs~\cite{chen2024zkml,kang2022scaling}, offering strong guarantees at high prover cost; \emph{peer ranking} (comparing outputs from multiple independently assigned nodes); \emph{benchmark injection} (hidden test queries with known answers mixed into normal traffic to catch low-effort nodes); and \emph{reputation weighting} based on previous node reliability.


In actuality, systems ought to use validation techniques according to the importance and danger of the activity under verification. While higher-value or security-sensitive workloads may support more robust cryptographic verification, low-cost inference requests can rely on less expensive processes like peer review, benchmark injection, and reputation weighting. By avoiding the need for costly zkML proofs for each task, this layered strategy enables the network to employ stronger guarantees when the extra prover overhead is worthwhile~\cite{bouchiha2024llmchain,conway2024opml,kang2022scaling,chen2024zkml,abbaszadeh2024zero}.




\subsection{Economic Layer and Closed-Loop Definition}

The economic layer (Fig.~\ref{fig:loop}) links user payments, token
issuance, and node compensation in a single cycle.
The same token mediates consumption and compensation, in contrast to
purely emission-driven models where token value depends on future-use
expectations rather than realized
demand~\cite{buterin2014next,nakamoto2008bitcoin}.

\begin{figure}[H]
\centering
\includegraphics[width=0.9\columnwidth]{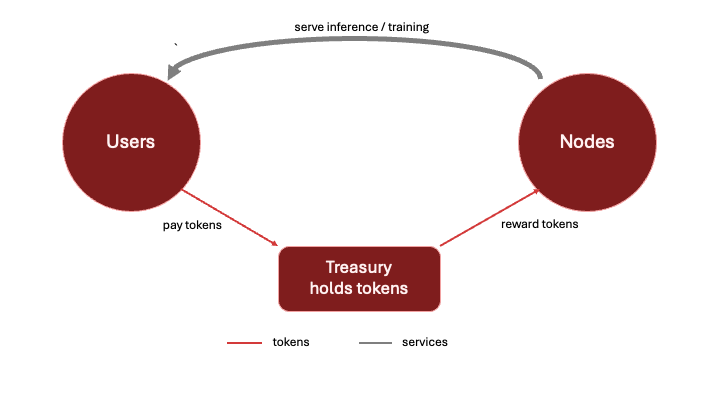}
\caption{Closed-loop token economy: users pay tokens into the treasury,
which rewards productive nodes; nodes provide AI services back to users.
Token value tracks real service demand rather than speculative expectations.}
\label{fig:loop}
\end{figure}
\FloatBarrier

\begin{definition}[Closed-loop AI economy]
\label{def:loop}
Let $\mathcal{N}$ be a decentralized AI network with native token $\tau$.
Let $C_\tau(t)$ denote the fraction of aggregate node compensation paid in
$\tau$ over a rolling window $W$ ending at time $t$, and let $I_w(t)$
denote the fraction of newly issued $\tau$ allocated according to
validated productive work over the same window.
We say $\mathcal{N}$ is a
\emph{$(\theta_c,\theta_w,W)$-closed-loop AI economy} if, for all
sufficiently large $t$:
\begin{enumerate}[label=(\roman*),leftmargin=*]
  \item $\tau$ is the only accepted on-network medium of payment for AI
        services;
  \item $C_\tau(t) \geq \theta_c$; and
  \item $I_w(t) \geq \theta_w$.
\end{enumerate}
Unless otherwise stated, we use $\theta_c = 0.75$, $\theta_w = 0.75$,
and $W = 90$\,days as default thresholds.
A system satisfying all three conditions is \emph{strictly closed-loop};
satisfying~(i) but missing at least one threshold is \emph{weakly
closed-loop}; failing~(i) is \emph{not closed-loop}.
\end{definition}

\section{Inference and Training: A Comparison}
\label{sec:inftrain}

Inference and training stress a decentralized substrate in fundamentally
different ways.
An inference request is \emph{stateless}: a fixed model $f_\theta$ maps
input $x$ to output $y = f_\theta(x)$, with narrow latency budgets and
straightforward replication.
With $k$ independent replicas each honest with probability $h > 1/2$,
majority-vote error decreases exponentially in $k$ by Hoeffding's
inequality~\cite{hoeffding1963probability}.
Batching and caching improve throughput but introduce a free-rider attack
surface: a node might serve cached results while claiming full computation
rewards, so the validation layer must distinguish genuine work from replay
via nonce tokens or fresh test inputs.

Training is \emph{stateful}: a gradient step
$\theta_{t+1} = \theta_t - \eta\,\widehat{g}_t$
modifies shared global parameters.
Decentralized training faces synchronization overhead (slow nodes
bottleneck SGD), high bandwidth (gradient payloads of several GB per
step), and \sloppy limited gradient verifiability.
Communication-efficient methods such as DiLoCo~\cite{douillard2023diloco}
and SWARM~\cite{ryabinin2023swarm} reduce bandwidth costs;
Byzantine-tolerant aggregators such as Krum~\cite{blanchard2017machine} and
coordinate-wise trimmed mean~\cite{yin2018byzantine} bound adversarial
influence.
We summarize the contrast in Table~\ref{tab:inftrain}.

\begin{table}[H]
\caption{Contrasting demands of decentralized inference and training.}
\label{tab:inftrain}
\centering
\footnotesize
\renewcommand{\arraystretch}{1.3}
\setlength{\tabcolsep}{4pt}
\begin{tabularx}{\columnwidth}{@{}>{\raggedright\arraybackslash}p{0.15\columnwidth}
  >{\raggedright\arraybackslash}p{0.30\columnwidth}
  >{\raggedright\arraybackslash}p{0.43\columnwidth}@{\hspace{0.06\columnwidth}}}
\toprule
\textbf{Property} & \textbf{Inference} & \textbf{Training} \\
\midrule
State        & Stateless                & Stateful parameters \\
Duration     & Milliseconds--seconds    & Hours to weeks \\
Coordination & Independent workers      & Synchronous/semi-synchronous \\
Bandwidth    & Low (request/response)   & High (gradient traffic) \\
Verification & Peer comparison feasible & Specialized checks needed \\
Reward       & Per query                & Per epoch/checkpoint \\
Examples     & opML~\cite{conway2024opml}, zkML~\cite{kang2022scaling},
               LLMChain~\cite{bouchiha2024llmchain} &
               Gensyn~\cite{gensyn2022}, DiLoCo~\cite{douillard2023diloco},
               SWARM~\cite{ryabinin2023swarm} \\
\bottomrule
\end{tabularx}
\end{table}
\FloatBarrier

A credible decentralized AI economy must accommodate both the training
that produces the models inference serves and the inference that generates
the revenue to fund further training.

\section{Incentive Design}
\label{sec:incentives}

\subsection{Threat Landscape}

We note the main incentive risks as below: 
\begin{enumerate}
    \item \emph{Free riding}: returning cached or trivial outputs while claiming full compensation,
    \item \emph{Sybil farming}: creating many shallow identities to capture disproportionate rewards,
    \item \emph{Collusion}: validators and provers coordinating to approve incorrect work, and
    \item \emph{Model plagiarism}: serving a copied model while claiming ownership.
\end{enumerate}

\subsection{Reward Function and Honest-Effort Condition}

For a node answering task $x$ with output $y$, define:
\begin{itemize}[leftmargin=*]
  \item $Q(y,x) \in [0,1]$: validator-aggregated quality score (average of
        $k$ independent ratings; under $f < 1/3$ Byzantine stake, at least
        $\lceil 2k/3 \rceil$ ratings are honest).
  \item $C(y) \geq 0$: verifiable compute-cost proxy (FLOPs attested by an
        attested runtime and checked against a reference forward pass of
        the declared model).
        Honest execution incurs the declared model's reference cost; a
        deviation that substitutes cheaper computation reports a lower
        $C(y)$, producing the gap $c_H - c_L$ in
        Proposition~\ref{prop:ic}.
  \item $P(y,x) \in \{0,S\}$: slashing term equal to staked bond $S$
        if a dispute against output $y$ succeeds within the challenge
        window, and $0$ otherwise.
\end{itemize}

The expected net reward to a node is
\begin{equation}
  \mathbb{E}[R(y,x)]
  \;=\; \alpha\, Q(y,x) \;-\; \beta\, C(y)
  \;-\; \gamma\, \mathbb{E}[P(y,x)],
  \label{eq:reward}
\end{equation}
where $\alpha > 0$ is the protocol reward rate per unit quality,
$\beta \geq 0$ is the node's per-FLOP cost (set by hardware and energy
prices, not a protocol lever), and $\gamma \in (0,1]$ is the slash
fraction.

\begin{proposition}[Sufficient stake for honest effort]
\label{prop:ic}
Under Assumption~\ref{assm:threat} and reward function~\eqref{eq:reward},
let honest execution yield expected quality $q_H$ at compute cost $c_H$,
and let the most profitable deviation yield $q_L \leq q_H$ at
$c_L \leq c_H$.
Suppose honest outputs are never successfully challenged, while deviations
are caught with probability at least $p \in (0,1]$.
Honest execution is a best response whenever
\begin{equation}
  S \;\geq\; \frac{\beta(c_H - c_L) - \alpha(q_H - q_L)}{\gamma\,p}.
  \label{eq:stakebound}
\end{equation}
If $\alpha(q_H - q_L) \geq \beta(c_H - c_L)$, the right-hand side is
non-positive and any $S \geq 0$ suffices.
\end{proposition}

\begin{proof}
Honest execution yields $\mathbb{E}[R_H] = \alpha q_H - \beta c_H$.
The most profitable deviation yields
$\mathbb{E}[R_L] = \alpha q_L - \beta c_L - \gamma p S$,
where $\mathbb{E}[P] = pS$.
Requiring $\mathbb{E}[R_H] \geq \mathbb{E}[R_L]$ gives
$\gamma p S \geq \beta(c_H - c_L) - \alpha(q_H - q_L)$;
dividing by $\gamma p > 0$ yields~\eqref{eq:stakebound}. \qed
\end{proof}

Equation~\eqref{eq:stakebound} captures the core trade-off.
Stronger detection (larger $p$) or a wider quality gap between honest and
low-effort work (larger $q_H - q_L$) each reduce the required bond.
When the quality loss from cheating already exceeds the compute savings,
the right-hand side is non-positive and even $S = 0$ satisfies the
inequality.
Conversely, when verification is weak, the required stake rises, which
protects the system but also raises the barrier for honest participants.

\section{Quantum Resilience and Post-Quantum Security}
\label{sec:quantum}

Quantum computing threatens blockchain architectures in two distinct ways,
with very different implications for a useful-work system.

\subsection{Shor's Algorithm and the Signature Layer}

\emph{Shor's algorithm}~\cite{shor1994algorithms} factors integers and
computes discrete logarithms in polynomial time, breaking the hardness
assumptions behind ECDSA, ECDH, and related classical signature
schemes~\cite{bernstein2025post,fernandez2020towards}.
With a fault-tolerant quantum computer of sufficient scale, an attacker
could derive private keys from exposed public keys and forge transactions on
any chain relying on elliptic-curve cryptography.
This vulnerability is not specific to this proposal; it affects all
classical blockchain designs.

\sloppy NIST finalized three post-quantum cryptographic standards in August 2024~\cite{nist2024pqc}: ML-KEM (CRYSTALS-Kyber, key encapsulation), ML-DSA (CRYSTALS-Dilithium, lattice-based signatures), and SLH-DSA (SPHINCS+, hash-based signatures). In the post-quantum era, researchers started focusing on quantum-based machine learning modeling~\cite{vhaduri2026we}, where a strong algorithm can play a vital role. ML-DSA and SLH-DSA are the recommended starting points for the transaction layer of any system targeting quantum resilience.

\subsection{Grover's Algorithm and Hash-Based Work}

\emph{Grover's algorithm}~\cite{grover1996fast} provides a quadratic
speedup for unstructured search.
Against an $n$-bit hash target, it reduces preimage-search cost from
$2^n$ to approximately $2^{n/2}$, halving the effective security level.
For Bitcoin's 256-bit target, this yields roughly 128-bit post-quantum
security---serious but recoverable.
Increasing hash output length restores the classical security margin at
modest verification cost.
Fig.~\ref{fig:qubit} illustrates the key distinction between classical
bits and qubits that underpins both speedups.

\begin{figure}[H]
\centering
\includegraphics[width=0.6\textwidth]{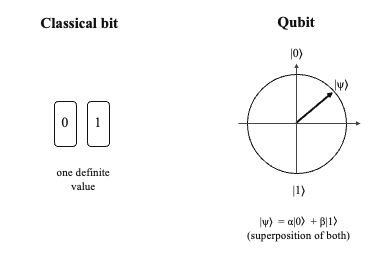}
\caption{Classical bit versus qubit.
A classical bit is fixed at $0$ or $1$; a qubit occupies the superposition
$|\psi\rangle = \alpha|0\rangle + \beta|1\rangle$, enabling the quantum
parallelism that underpins Grover's and Shor's speedups.}
\label{fig:qubit}
\end{figure}
\FloatBarrier

\subsection{Why ML Work Layers Are Less Exposed}

The key observation is that model inference and gradient computation are
\emph{not} unstructured search problems.
Grover's speedup applies to functions of the form
$f : \{0,1\}^n \to \{0,1\}$ where the goal is to find an input $x$ with
$f(x) = 1$.
The useful-work tasks in our architecture: neural-network forward passes
and gradient aggregations consist of dense matrix multiplications,
non-linear activations, and normalization operations.
These do not admit the ``query one candidate'' structure that Grover
exploits; there is no oracle-queryable predicate to invert.

This means that replacing hash puzzles with ML computation structurally
avoids making artificial preimage search the central economic activity of
the network, and avoids giving a quantum adversary the specific
computational structure Grover requires.
The argument is architectural rather than a formal cryptographic reduction;
a post-quantum security game for ML-work consensus remains an open problem.

\subsection{Layer-by-Layer Assessment}

Table~\ref{tab:quantum} summarizes the quantum threat profile and
recommended mitigations at each layer of the architecture.

\begin{table}[H]
\caption{Quantum threat profile and post-quantum mitigations by
architectural layer.}
\label{tab:quantum}
\centering
\footnotesize
\renewcommand{\arraystretch}{1.3}
\setlength{\tabcolsep}{3pt}
\begin{tabularx}{\columnwidth}{@{}
  >{\raggedright\arraybackslash}p{0.20\columnwidth}
  >{\raggedright\arraybackslash}p{0.19\columnwidth}
  >{\raggedright\arraybackslash}p{0.24\columnwidth}
  >{\raggedright\arraybackslash}p{0.31\columnwidth}@{}}
\toprule
\textbf{Layer} & \textbf{Classical risk} & \textbf{Quantum threat}
  & \textbf{Recommended mitigation} \\
\midrule
Transaction (signatures) &
  Key forgery &
  Shor: poly-time ECDSA break &
  ML-DSA, SLH-DSA~\cite{nist2024pqc} \\
Hash PoW &
  Brute-force ($2^n$) &
  Grover: $2^{n/2}$ effective cost &
  Larger hash output; or ML-work replacement \\
ML work layer &
  Free-riding, model substitution &
  Not accelerated by Grover &
  Attestation, zkML~\cite{kang2022scaling} \\
zkML proofs &
  Pairing-based hardness &
  Shor threatens EC pairings &
  STARK-style proofs (hash-only)~\cite{ben2018scalable} \\
\bottomrule
\end{tabularx}
\end{table}
\FloatBarrier

At the \textbf{transaction layer}, ML-DSA and SLH-DSA migration can
proceed incrementally as wallets and clients update.
At the \textbf{validation layer}, hash-based commitments can be
strengthened by increasing output length, and STARK-based proof systems
are strong post-quantum candidates because they rely on
collision-resistant hashing rather than elliptic-curve
pairings~\cite{ben2018scalable}.
For zkML specifically, pairing-based SNARKs (e.g., Groth16) face
Shor-algorithm risk; STARK-style proofs avoid this
exposure~\cite{kang2022scaling,chen2024zkml}.
At the \textbf{ML work layer} itself, the dense linear algebra of
inference and backpropagation does not benefit from Grover-style search,
yielding a structural quantum-resilience advantage over hash-based PoW
that does not depend on parameter choices.

\section{Related Systems}
\label{sec:related}

None of the deployed systems fully combines inference, training, and
closed-loop token demand.
Bittensor~\cite{rao2020bittensor} has advanced peer-ranked intelligence
markets but its token flows are not directly defined by the metered
service-payment criterion in Definition~\ref{def:loop}; it is therefore
at most weakly closed-loop.
Gensyn~\cite{gensyn2022} focuses on decentralized training with
proof-of-learning validation but does not describe a general inference
market.
Akash~\cite{akash2018whitepaper} offers the broadest compute marketplace
in the comparison but leaves AI-specific correctness to clients.
opML~\cite{conway2024opml} and zkML~\cite{kang2022scaling} address
inference verification but not the token-economy or training layers.
LLMChain~\cite{bouchiha2024llmchain} adds a reputation system but does
not target closed-loop economics.
Table~\ref{tab:related} summarizes these comparisons.

\begin{table}[H]
\caption{Comparison of representative decentralized AI and useful-work
systems. \emph{Loop} denotes whether the system satisfies, partially
satisfies, or does not target the closed-loop criterion in
Definition~\ref{def:loop}.}
\label{tab:related}
\centering
\footnotesize
\renewcommand{\arraystretch}{1.3}
\setlength{\tabcolsep}{4pt}
\begin{tabularx}{\columnwidth}{@{}
  >{\raggedright\arraybackslash}p{0.18\columnwidth}
  >{\centering\arraybackslash}p{0.11\columnwidth}
  >{\centering\arraybackslash}p{0.12\columnwidth}
  >{\raggedright\arraybackslash}p{0.23\columnwidth}
  >{\centering\arraybackslash}p{0.13\columnwidth}
  >{\centering\arraybackslash}p{0.09\columnwidth}@{}}
\toprule
\textbf{System} & \textbf{Inference} & \textbf{Training}
  & \textbf{Validation} & \textbf{Loop} & \textbf{Ref.} \\
\midrule
Bitcoin             & --          & --
  & Hash puzzle       & n/a      & \cite{nakamoto2008bitcoin} \\
Bittensor           & \checkmark  & partial
  & Peer ranking      & partial  & \cite{rao2020bittensor} \\
Gensyn              & --          & \checkmark
  & PoL / optimistic  & partial  & \cite{gensyn2022} \\
Akash               & generic     & generic
  & Client-defined    & --       & \cite{akash2018whitepaper} \\
opML                & \checkmark  & --
  & Optimistic        & --       & \cite{conway2024opml} \\
LLMChain            & \checkmark  & --
  & Reputation        & --       & \cite{bouchiha2024llmchain} \\
zkML                & \checkmark  & --
  & Cryptographic     & n/a      & \cite{kang2022scaling} \\
\textbf{This paper} & \checkmark  & \checkmark
  & Hybrid            & \checkmark & -- \\
\bottomrule
\end{tabularx}
\end{table}
\FloatBarrier



There are various ways that the verification tools vary at a significant scale. For example, while opML is cheap, it adds dispute latency and while zkML offers higher correctness, prover overhead is very high for for large transformer
models~\cite{kang2022scaling,chen2024zkml}. Similarly, reputation systems are statistical, although these are very practical systems. For any production system, these tools are jointly used depending on their risk and significance trade-offs. Additionally, weak relationship exist between token value and metered demand, jeopardizing condition (ii) defined in Definition~\ref{def:loop}, when compensation is derived from speculations~\cite{mafrur2025ai}.

\section{Open Challenges and Future Directions}
\label{sec:future}

Four challenges require further work. First, zkML prover costs remain prohibitive for routine inference at commodity prices~\cite{kang2022scaling,chen2024zkml}; until they fall substantially, most verification will rely on redundant execution or optimistic dispute windows~\cite{conway2024opml,kalodner2018arbitrum}. It also needs exploration of how this zkML prover can be applicable to have provable security-bound in distributed computing, specifically, involving hardware and IoT devices~\cite{vhaduri2023bag}. Second, decentralized training needs a cleaner connection between communication-efficient methods, Byzantine-robust aggregation, and token rewards that incentivize convergence rather than only completed compute cycles~\cite{douillard2023diloco,ryabinin2023swarm}. Third, governance and data provenance, training data licensing, model ownership, data privacy in split learning, adversarial and privacy vulnerabilities, and harmful-output liability remain unresolved, and determine whether participants trust the network enough to use it at scale~\cite{dibbo2024improving}. Similar challenges arise in real-world AI deployments, where model performance
can vary substantially across heterogeneous operating environments, requiring adaptive and trustworthy learning mechanisms~\cite{vhaduri2023environment}. Fourth, and most directly relevant to this venue, the quantum-resilience claim in Section~\ref{sec:quantum} is an architectural argument, not a cryptographic proof. A formal post-quantum security game for ML-work consensus, including an exact analysis of validation primitives, a quantum reduction for the honest-majority assumption under Grover-accessible oracles, and an analysis of lattice-signature integration, remains an open research problem.

\section{Conclusion}
\label{sec:conclusion}

This paper proposed a decentralized AI economy where inference and training, rather than hashing, become the productive work of the network. The three-layer architecture, closed-loop token criterion, and incentive-compatibility result together identify three necessary conditions for viability: verification must make low-quality work detectable, incentives must make honest work individually rational, and real service demand must support rewards or the token economy collapses into speculative emissions. The quantum-resilience analysis adds a fourth observation of particular relevance for emerging quantum technologies. Replacing hash puzzles with ML computation shifts the work layer away from the preimage-search structure that Grover's algorithm exploits; this is a structural advantage, not merely a parameter adjustment. Signature-layer exposure to Shor's algorithm is a shared liability of all classical blockchains, but it is addressable through standardized post-quantum signatures (ML-DSA, SLH-DSA) and hash-based proof systems (STARKs) at the validation layer. Designed well, useful work networks can provide an open substrate for AI compute that is simultaneously more economically efficient and more structurally resilient to quantum adversaries than classical proof-of-work, making them a natural target for formal analysis in the quantum computing and quantum machine learning research community.

\bibliographystyle{splncs04}
\bibliography{ref_revised}

\end{document}